\documentclass[reqno]{amsart}
\usepackage{amsmath,amsthm,amsfonts,amssymb}
\RequirePackage{graphicx}
\numberwithin{equation}{section}
\theoremstyle{plain}

\newtheorem{theorem}{Theorem}

\newtheorem{corollary}[theorem]{Corollary}

\newtheorem{lemma}[theorem]{Lemma}

\newtheorem{assumption}[theorem]{Assumption}

\begin{document}

\title[Long Range Lattice Gas Models]{Lattice Gas Models with Long Range Interactions}

\author{David Aristoff}
\address
{Department of Mathematics \newline
\indent Colorado State University \newline
\indent 1874 Campus Delivery \newline
\indent Fort Collins, CO-80523 \newline
\indent United States of America}
\email{
aristoff@math.colostate.edu}

\author{Lingjiong Zhu}
\address
{Department of Mathematics \newline
\indent Florida State University \newline
\indent 1017 Academic Way \newline
\indent Tallahassee, FL-32306 \newline
\indent United States of America}
\email{
zhu@math.fsu.edu}

\date{31 October 2016.}


\begin{abstract}
We study microcanonical lattice gas models with long range 
interactions, including 
power law interactions. 
We rigorously obtain a variational 
principle for the entropy. 
In a one dimensional example, 
we find a first order phase transition 
by proving the entropy is 
non-differentiable along a 
certain curve. 
\end{abstract}

\maketitle
\section{Introduction}

In this article we study lattice gas models with certain long range pair interactions. 
Our models are generalizations of 
certain mean field and random graph 
models, in which all sites or nodes interact with all 
others with equal strength. 
In contrast with mean field models, 
we allow the interaction strength 
to decay, but at a rate sufficiently 
slow that interactions between 
far away sites are still significant.

Our models can be described as follows. 
Consider the cubic lattice in $d$ dimensions, 
rescaled so that the spacing 
between adjacent sites is $1/n$. 
We consider configurations of 
particles on the $\sim n^d$ 
lattice sites 
that fit inside the  
$d$-dimensional unit 
cube ${\mathcal C}=[0,1]^d$. 
Each configuration consists 
of an arrangement of particles 
on the lattice sites in ${\mathcal C}$. 
We say a site is occupied if there 
is a particle there;  
each lattice site is 
either occupied 
by one particle or unoccupied. 
Configurations are assigned an 
energy from a pair potential 
$\psi:[0,\sqrt{d}] \to {\mathbb R}$. 
Interactions are only between 
occupied sites. 
The energy associated with 
distinct occupied sites
$x$ and $y$ is $\psi(|x-y|)$, and   
the total energy of a configuration 
is obtained by summing $\psi(|x-y|)$ 
over all occupied sites $x,y$. 
We assume there is $0<r<d$ such
that $\psi(t)$ 
diverges at least as slowly as $t^{-r}$ 
as $t \to 0$.
As a consequence, 
the interactions 
between sites that are far apart 
(relative to the lattice spacing) 
make nontrivial contributions 
to the total energy.
Such interactions are 
sometimes called {\em long range}~\cite{barre_bouchet,bouchet,campa2,dauxois}. Equivalence of  
thermodynamic ensembles 
breaks down in this regime~\cite{barre,barre_mukamel}, 
so we will consider 
only the microcanonical ensemble, 
in which energy density and site 
occupancy density are held fixed.

When $\psi$ is a constant function, 
our model is a microcanonical 
version of the
Curie-Weiss mean field 
Ising model. When $d = 2$, our 
model is related to certain 
random graphs. This is because
the occupancy pattern on the lattice 
can be mapped to an adjacency matrix 
for a graph on $n$ nodes, where an occupied 
site is an edge in the graph, 
and an unoccupied site corresponds 
to the absence of an edge. In 
this case, $\psi$ corresponds 
to an interaction between graph edges. 
For discussion on the connection 
to the Curie-Weiss model and 
random graphs, see the comments in Section~\ref{sec:discussion} below.  

We now give a rough 
description of our main 
results. Each configuration 
can be described by an occupancy 
pattern -- a function 
with value $1$ at each occupied 
site and $0$ at each unoccupied 
site. More precisely, 
each configuration corresponds to  
an occupancy pattern $f$ defined on ${\mathcal C}$ 
as follows: if site $x$ is occupied (resp. unoccupied), $f$ is equal to $1$ (resp. $0$) 
on a 
$d$-dimensional cube of side length $1/n$ 
centered at $x$. The associated particle 
density is $\int_{\mathcal C} f(x)\,dx$ and the energy 
density can be approximated by
\begin{equation*}
n^d\int_{{\mathcal C}^2} f(x)f(y)\psi(|x-y|)\,dx\,dy.
\end{equation*} 
The potential $\psi$ 
is divided by $n^d$ so 
that  
the energy density 
scales appropriately as $n \to \infty$. 
In this limit, there is a continuum of lattice sites in ${\mathcal C}$, and $f$ becomes an occupancy
density function with values in $[0,1]$. 
One can imagine that such $f$ is obtained 
by smoothing out occupancy pattern 
functions for large finite~$n$, 
with $f(x)$ representing 
the probability that site $x$ 
is occupied. The entropy density
associated with the occupancy density 
$f$ is
\begin{equation}\label{entropy_density}
- \int_{{\mathcal C}}\left[f(x)\log f(x) + 
(1-f(x))\log(1-f(x))\right]\,dx.
\end{equation}
The entropy density is simply 
the log of the 
number of configurations with 
density approximately given by~$f$, 
normalized by $n^d$. 
The above heuristics show that as 
$n\to \infty$, the energy density, 
entropy density, and particle 
density can be accurately expressed in 
terms of the occupancy density 
$f:{\mathcal C} \to {\mathbb R}$. 
If these heuristics are correct, 
one expects that as $n\to \infty$, at fixed energy and particle density, 
configurations will have  
occupancy densities that 
approach optimizers $f_*$ of the
entropy density~\eqref{entropy_density}
subject to the constaints on 
particle and energy density. 
It 
is then  
straightforward to write 
the corresponding
Euler-Lagrange equations and, 
in principle at least, find the optimizers $f_*$.
For details on the 
above ideas from the point 
of view of large deviations 
theory, see for instance~\cite{barre_bouchet}. 

The main contribution of this 
paper is to show that the 
above heuristics can be 
made mathematically rigorous 
under very weak assumptions 
on the potential energy $\psi$. 
(Though we focus on lattice gas 
models, our arguments easily 
adapt to other long 
range interacting models, 
{\em e.g.} the 
$\alpha$-Ising model~\cite{barre}. 
See also~\cite{Mori0} for a closely related result.)
We also give an example 
of an interaction for 
which the model has a 
first order phase transition. 
To our knowledge, such transitions 
had not yet been rigorously demonstrated in 
microcanonical models of this type.

Rigorous results in 
the long range setting 
described above are relatively 
scarce.
We mention that similar rigorous
results have been proved 
in microcanonical spin 
models~\cite{Mori0}
and in the grand canonical Ising model with 
Kac interactions~\cite{benois1,benois,Vollmeyer}. 
See also~\cite{barre,campa}
for similar work on 
the $\alpha$-Ising model, 
and~\cite{mori1,mori2} for 
studies of more general long 
range interacting Ising models.
For rigorous analysis of other
mean field type 
models, see~\cite{benois,costenuic,ellis,kac}. 
We also mention related work 
on random graph 
models in which the interaction
depends on the number of 
edges and other subgraphs; 
see~\cite{aristoff_zhu,
diaconis,kenyon,kenyon2,radin_ren_sadun,radin_sadun,
radin_sadun2}.

This article is organized as follows. In Section~\ref{sec:models} 
we describe our models in detail. 
We
present 
a variational principle for the 
entropy and the corresponding 
Euler-Lagrange equations 
in Section~\ref{sec:large_dev}. 
Using these results, 
we show a phase transition 
in a one dimensional model 
in Section~\ref{sec:example}. 
In Section~\ref{sec:discussion}, 
we discuss connections between 
our models and certain mean field and 
random graph models. 
All proofs are in Section~\ref{sec:appendix}. 

\section{Notation and assumptions}\label{sec:models}

Fix a dimension $d \ge 1$, and 
for $n\ge 1$ define the lattice 
$\Lambda_n = \{1,\ldots,n\}^d$. 
Lattice sites (that is, elements of $\Lambda_n$) will be denoted 
by $I,J$. Each lattice site $I$ 
can be occupied or not. 
A particle configuration is 
an assignment of occupancy to 
each site. More precisely, 
a particle configuration is a 
function
$\eta:\Lambda_n \to \{0,1\}$. 
(We sometimes write $\eta_n$ 
to emphasize dependence 
on $n$.) Here, $\eta(I) = 1$ if 
site $I$ is occupied, and $\eta(I) = 0$ 
otherwise. Recall 
$\psi:[0,\sqrt{d}] 
\to {\mathbb R}$ is a given 
pair potential. The interaction between 
sites $I,J \in \Lambda_n$ 
is defined by 
\begin{equation}\label{phin}
\phi_n(I,J) := \psi\left(n^{-1}|I-J|\right),
\end{equation}
where $|\cdot|$ is 
the usual Euclidean norm in ${\mathbb R}^d$.
Recall ${\mathcal C} = [0,1]^d$ 
is the $d$-dimensional unit cube. 
Let ${\mathcal C}_I$
be a $d$-dimensional cube of 
side length $1/n$ centered 
(approximately) at $n^{-1}I$. 
More precisely,
${\mathcal C}_I= \{x \in {\mathcal C}\,:\,I-1 \le nx < I\}$, 
where $1$ represents the all 
ones vector and the inequalities 
are componentwise. 
Throughout, we will associate 
a particle configuration $\eta$
to a occupancy density function 
$f^\eta$ obtained by setting 
$f^\eta$ equal to $1$ on cubes ${\mathcal C}_I$
corresponding to occupied sites $I$, 
and $0$ otherwise. 
More precisely, $f^\eta:{\mathcal C} 
\to [0,1]$ is defined by 
\begin{equation*}
f^\eta(x) = \eta(I), \qquad \text{if }x \in {\mathcal C}_I.
\end{equation*}
See Figure~\ref{fig1}. Let ${\mathbb P}_n$ be 
the uniform probability measure 
on particle configurations, 
\begin{equation*}
{\mathbb P}_n(\eta) = 2^{-n^d}, \qquad \text{for all }\eta:\Lambda_n \to \{0,1\}.
\end{equation*}
This defines an equivalent
measure on occupancy density functions. 
That is, under the map $\eta \to f^\eta$, ${\mathbb P}_n$ 
pushes forward to a probability measure 
on the space of measurable 
functions ${\mathcal C} \to [0,1]$. We denote this measure by
the same symbol, 
${\mathbb P}_n$, since no confusion 
should arise. Define 
the {energy density} $E_n$  
of $\eta:\Lambda_n \to \{0,1\}$ as 
the sum of $\phi_n(I,J)$ over all 
pairs of occupied sites $I,J$, appropriately normalized:
\begin{equation}\label{E}
E_n(\eta) = n^{-2d}\sum_{I,J \in \Lambda_n} 
\eta(I) \eta(J) \phi_n(I,J). 
\end{equation}
Define also the particle density $N_n$ 
of $\eta:\Lambda_n \to \{0,1\}$ as 
the average site occupancy:
\begin{equation}\label{N}
N_n(\eta) = n^{-d}\sum_{I \in \Lambda_n} \eta(I).
\end{equation}
Fix parameters $\xi, \rho\in {\mathbb R}$, and
define the microcanonical {entropy}
\begin{equation}\label{entropy}
S(\xi,\rho) = \lim_{\delta \to 0^+}
\lim_{n\to \infty} n^{-d} \log {\mathbb P}_n\left(E_n(\eta) \in (\xi-\delta,\xi+\delta), \,N_n(\eta) \in (\rho-\delta,\rho+\delta)\right).
\end{equation}
We show below the limit 
defining $S(\xi,\rho)$ exists under the 
following assumption.

\begin{figure}
\begin{center}
\vskip-120pt
\includegraphics[scale=0.55]{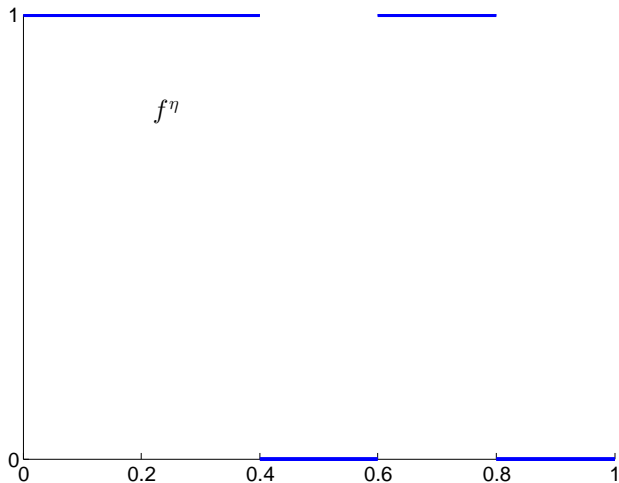}
\vskip-120pt
\caption{The graph of 
the function $f^\eta$ 
when $d=1$, $n = 5$,  
$\eta(1) = \eta(2) = \eta(4)= 1$ and $\eta(3)=\eta(5) = 0$.}\label{fig1}
\end{center}
\end{figure}

\begin{assumption}\label{A1}
The map 
$(x,y) \mapsto \psi(|x-y|)$ is 
in $L^q({\mathcal C}^2)$ 
for some $q > 1$. Moreover, 
it is Riemann 
integrable. 
\end{assumption}
Assumption~\ref{A1} will hold 
throughout the remainder of the paper. 
As a typical example 
of interactions satisfying 
this assumption, we keep in mind the case 
of power law interactions in 
dimension $d =1$, where $\psi(t) = t^{-r}$ for 
$t \in (0,1]$ and $\psi(0) = 0$, 
with $r \in (0,1)$ constant. (We 
set $\psi(0) = 0$ simply so that 
a particle does not interact 
with itself.) We note, however,  
that the interaction need 
not even be continuous. 
In Section~\ref{sec:example} below, 
we consider a modified 
power law interaction for 
which the entropy $S$ is singular.

Before proceeding we comment 
on boundary conditions. Note 
that the definitions~\eqref{E}-~\eqref{N} 
suggests free boundary conditions. 
However, we note that periodic 
boundary conditions can be obtained 
by replacing Euclidean 
distance $|\cdot|$ in~\eqref{phin} with 
distance on the flat torus ${\mathbb R}^d/{\mathbb Z}^d$. 
In dimension $d =1$, this corresponds 
to setting $\psi(t) = \psi(1-t)$ 
for all $t \in [0,1]$. 
See~\cite{barre} for a 
similar discussion 
of boundary conditions in 
the $\alpha$-Ising model.

\section{Large deviations, entropy, 
and Euler-Lagrange equations}\label{sec:large_dev}

Before stating our results 
we give a sketch of the 
arguments. First, we show that 
the formula~\eqref{entropy_density} 
represents the logarithm of the 
number of configurations with 
occupancy density approximately 
equal to $f$. The relevant result  
is Theorem~\ref{thm_large deviation principle} below.
Roughly, for 
large $n$ and
suitable collections $A$ of 
particle configurations, we show that
\begin{align}\begin{split}\label{large deviation principle}
&n^{-d}\log {\mathbb P}_n(f^\eta \in A) \approx \\
&\qquad \sup_{f \in A} \left\{-\log 2 - \int_{{\mathcal C}} \left[f(x)\log f(x)+(1-f(x))\log(1-f(x))\right]\,dx\right\}.
\end{split}
\end{align}
The extra term $\log 2$ comes from the probability 
normalization. The equation~\eqref{large deviation principle} follows from 
a {\em large deviations principle}, and the 
quantity in brackets (multiplied by $-1$) is called 
the {\em rate function}.
See below for precise definitions 
of this terminology.

The trick is to prove 
the large deviations principle in a topology strong 
enough so that the set  
of suitable collections contains 
the collection $A$ we are interested in.  
Since we want to compute the 
microcanonical 
entropy~\eqref{entropy}, 
we take $A$ to be the collection
of configurations with energy 
density and particle density
approximately equal to $\xi$ and $\rho$, 
respectively. 
If the (approximate) energy density 
$\int_{{\mathcal C}^2} f(x)f(y)\psi(|x-y|)\,dx\,dy$ and particle density 
$\int_{\mathcal C}f(x)\,dx$ 
are continuous in $f$, then $A$ is 
indeed suitable and we can use~\eqref{large deviation principle} 
to compute the microcanonical 
entropy~\eqref{entropy}. This 
is a consequence of the well-known 
{\em contraction principle}~\cite{dembo,VaradhanII} in 
large deviations theory. (We 
state the contraction principle 
in  
Section~\ref{sec:appendix} below.)
It turns out that particle density 
is continuous in any reasonable topology, 
but it is not trivial to 
show energy density is continuous
in a suitable topology. 
See Lemma~\ref{lem_cts} below.
There are some additional 
technical issues associated with 
showing the energy density can be 
well approximated by $\int_{{\mathcal C}^2} f(x)f(y)\psi(|x-y|)\,dx\,dy$ in 
the sense of exponential 
equivalence~\cite{VaradhanII} 
(this term is defined 
in Section~\ref{sec:appendix}).
See Lemma~\ref{lem_equiv} below.

We are now ready to state our results. 
For notational convenience, we write
\begin{equation*}
H_{\textup{bin}}(t) = \begin{cases}t \log t + (1-t)\log(1-t)+\log 2, & t \in [0,1]\\
\infty, & t \notin [0,1]\end{cases}.
\end{equation*}
(We write $H_{bin}$ here because  
of the similarity to 
the binary entropy function 
$-t\log_2 t  -(1-t)\log_2 (1-t)$.)
Before proceeding, we introduce 
the terminology we need from large 
deviations theory. A sequence 
${\mathbb Q}_n$ of probability
measures on a topological space 
${\mathcal T}$ is said to satisfy a {large deviations principle} 
with speed
$a_n$ and rate function 
$K : {\mathcal T} \to {\mathbb R}$ 
if $K$ is non-negative and lower semicontinuous, and for any
measurable set $A\subset {\mathcal T}$,
\begin{equation}\label{large deviation principle_ineq}
-\inf_{x \in A^\circ} K(x) 
\le \liminf_{n \to \infty} 
a_n^{-1}\log {\mathbb Q}_n(A) 
\le \limsup_{n \to \infty} 
a_n^{-1}\log {\mathbb Q}_n(A) 
\le -\inf_{x \in {\bar A}} K(x),
\end{equation}
where $A^\circ$ denotes the 
interior of $A$ and ${\bar A}$ 
the closure of $A$. We refer 
to the first inequality 
in~\eqref{large deviation principle_ineq} as the 
{lower bound} and the 
last inequality in~\eqref{large deviation principle_ineq} 
as the {upper bound}. Note 
that compared to the description 
above, we have replaced $\sup -K$ 
with $-\inf K$. This is 
so that we can be consistent with 
standard notations in large 
deviations theory.

Throughout we fix $p \in [1,\infty)$ 
and $q \in (1,\infty]$ with $p^{-1}+q^{-1}=1$. We 
prove a large deviations principle for ${\mathbb P}_n$ 
on the Banach space 
of functions in $L^p({\mathcal C})$ 
endowed with the weak 
topology. We denote 
this space by ${\mathcal X}$. 
We will also consider the 
subset ${\mathcal Y} = \left\{f \in {\mathcal X}\,:\, 
f(x) \in [0,1] \text{ for a.e. }x\right\}\subset {\mathcal X}$.
Unless otherwise specified, we endow
${\mathcal Y}$ with the 
subspace topology. 

\begin{theorem}\label{thm_large deviation principle}
The sequence ${\mathbb P}_n$ 
satisfies a large deviations 
principle on ${\mathcal X}$ with 
speed $n^d$ and rate function 
\begin{equation*}
H(f) = \int_{{\mathcal C}} H_{\textup{bin}}(f(x))\,dx.
\end{equation*}
\end{theorem}

Consider the 
following constrained subset of ${\mathcal Y}$, 
\begin{equation*}
{\mathcal Y}_{\xi,\rho} := \left\{
f \in {\mathcal Y}\,:\,
\int_{{\mathcal C}^2} 
f(x)f(y)\psi(|x-y|)\,dx\,dy = \xi, \,
\int_{\mathcal C}f(x)\,dx = \rho\right\}.
\end{equation*}
Abusing notation, we refer to 
$\int_{{\mathcal C}^2} 
f(x)f(y)\psi(|x-y|)\,dx\,dy$ 
as the energy density, 
even though when $f = f^\eta$ 
this expression is 
not exactly equal to~\eqref{E}. 
We show
in Lemma~\ref{lem_equiv} below 
that they are nonetheless close in the 
sense of exponential equivalence 
(this term is defined precisely 
above Lemma~\ref{lem_FL} in 
Section~\ref{sec:appendix} below).
Note that
$\int_{\mathcal C}f^\eta(x)\,dx$
is exactly equal to the particle density 
of $\eta$ defined in~\eqref{N}. 
Thus, we think of ${\mathcal Y}_{\xi,\rho}$ 
as the collection of 
occupancy density functions $f$
with energy density $\xi$ and particle 
density $\rho$. 
As discussed above, exponential equivalence,  Theorem~\ref{thm_large deviation principle}, and the contraction principle lead  
to the following variational 
expression for the entropy.
\begin{theorem}\label{thm_var}
We have
\begin{equation}\label{var_prob}
S(\xi,\rho) = -\inf_{f \in {\mathcal Y}_{\xi,\rho}} 
H(f) = \sup_{f \in {\mathcal Y}_{\xi,\rho}} \left[-
H(f)\right],
\end{equation}
with the infimum over the empty set 
equal to $\infty$ by convention.
\end{theorem}

We note that a very similar 
rigorous result was recently proved, 
using direct arguments,  
in~\cite{Mori0}. In our 
proof, we use 
the machinery of 
large deviations theory, 
proving a large deviations 
principle for ${\mathbb P}_n$ 
and using the contraction 
principle and exponential 
equivalence to get a 
variational principle 
for the entropy. 
Compared to the result 
in~\cite{Mori0}, our 
assumptions on the interaction $\psi$
are weaker; in particular, 
we allow for interactions 
$\psi$ that are non-smooth 
away from $0$. On the 
other hand, the article~\cite{Mori0}
considers different domain 
shapes as well as more general 
short range interactions. While 
many of our arguments could be 
generalized in this way, we do 
not pursue this direction, 
partly because 
of our interest in the connection 
of our problem with random graph 
models (in which a square domain 
represents an adjacency matrix).

Below we will refer to 
functions $f_* \in {\mathcal Y}_{\xi,\rho}$ with $S(\xi,\rho) = -H(f_*)$ 
as optimizers of the 
variational problem~\eqref{var_prob}. 
Optimizers represent the most 
likely structure of large 
particle configurations. 
For instance, if 
$f_*$ is the unique 
optimizer of~\eqref{var_prob} 
and $n$ is large, then 
$f_*( n^{-1}I)$ is roughly the 
probability that 
$\nu(I) = 1$, {\em i.e.}, there 
is a particle at site $I \in \Lambda_n$.

Standard results in 
the calculus of variations 
lead to the following. Whenever $(\xi,\rho)$ 
corresponds to acheivable values 
of energy and particle density, 
compactness arguments show that optimizers of the variational problem~\eqref{var_prob} exist; moreover
optimizers 
in the interior of the 
appropriate function space 
satisfy the Euler Lagrange equations. 
To make these statements precise, we 
define 
\begin{equation*}
 \Omega = 
\left\{(\xi,\rho)\,:\, {\mathcal Y}_{\xi,\rho} \ne \emptyset\right\},
\end{equation*}
as the region of achievable energy 
and particle densities, and write
\begin{equation*}
{\mathcal F} = \left\{f \in {\mathcal Y}\,:\, \exists \epsilon>0 \text{ s.t. }
f(x) \in [\epsilon,1-\epsilon] \text{ for a.e. }x\right\}
\end{equation*} 
for the interior of ${\mathcal Y}$ 
with respect to the essential sup norm.

\begin{theorem}\label{thm_EL}
Optimizers 
of~\eqref{var_prob} exist whenever 
$(\xi,\rho) \in \Omega$. 
If $
f_* \in {\mathcal F}\cap {\mathcal Y}_{\xi,\rho}$
is an optimizer, then for a.e. $x$, 
either
\begin{equation}\label{EL_rearrange}
f_*(x) = \frac{\exp\left(\mu + \beta\int_{\mathcal C} f_*(y)\psi(|x-y|)\,dy\right)}{1+\exp\left(\mu + \beta\int_{\mathcal C} f_*(y)\psi(|x-y|)\,dy\right)},
\end{equation}
for some $\beta,\mu \in {\mathbb R}$ 
with $(\beta,\mu)\ne (0,0)$, or 
\begin{equation}\label{EL_rearrange2}
\int_{\mathcal C} f_*(y)\psi(|x-y|)\,dy \equiv \xi/\rho.
\end{equation}
\end{theorem}

Viewing the expression in~\eqref{EL_rearrange} as 
a convolution leads to the following corollary.

\begin{corollary}\label{cor_cts}
If the Euler-Lagrange equation~\eqref{EL_rearrange} holds, 
then $f_*$ is continuous.
\end{corollary}

Intuitively, equation~\eqref{EL_rearrange2} 
holds only when the constraints 
take over in the optimization (see 
the discussion following 
Theorem 9.1 in~\cite{clarke}), so 
in ``most'' cases we expect 
that instead the Euler-Lagrange equations~\eqref{EL_rearrange} hold. 
Corollary~\ref{cor_cts} has 
the following interesting 
consequence. Suppose 
that $(\xi,\rho) \in \Omega$ 
and $\xi \ne \lambda \rho^2$, 
where $\lambda:=\int_{{\mathcal C}^2} \psi(|x-y|)\,dx\,dy$. 
Then the constant valued occupation 
density $f \equiv \rho$ cannot 
be an optimizer, since it has energy density
\begin{equation*}
\int_{{\mathcal C}^2} f(x)f(y)\psi(|x-y|)\,dx\,dy
= \rho^2 \int_{{\mathcal C}^2} \psi(|x-y|)\,dx\,dy = \lambda \rho^2 \ne \xi.
\end{equation*}
Thus, when $\xi \ne \lambda \rho^2$, 
any optimizer $f_*$ of $S(\xi,\rho)$
must be nonconstant. 
Suppose such $f_*$
satisfies the Euler-Lagrange 
equations~\eqref{EL_rearrange}. 
Then $f_*$ is continuous and 
nonconstant, say with 
two distinct values $a<b$, so 
Corollary~\ref{cor_cts} 
and the intermediate value 
theorem show that
$f_*$ takes every value in 
the interval $[a,b]$. 
In particular, 
$f_*$
cannot be constant or 
piecewise constant. 
Such optimizers $f_*$ have  
a spacially inhomogeneous 
occupation density profile. 
(See~\cite{barre} and Figure~\ref{fig2} below
for examples where the 
optimizer has a curved structure.) 
Note the contrast with
typical short range interactions, 
for which optimizers of the 
entropy have a spatially 
homogeneous (in 
pure phases) or piecewise 
homogenous (in 
mixed phases) density profile 
as system size goes to infinity. 

It is also interesting to consider 
the case of constant valued interactions. 
Suppose $\psi \equiv \lambda$ is constant. Then one of the constraints 
is redundant: if $\int_{\mathcal C} f(x)\,dx = \rho$ then 
$$\int_{{\mathcal C}^2}f(x)f(y)\psi(|x-y|)\,dx\,dy = \lambda \rho^2.$$ Thus, 
particle density $\rho$ completely 
determines energy density $\xi$ 
via $\xi = \lambda \rho^2$. 
In this case, it is easy 
to see that the only optimizer of 
the entropy 
$S(\lambda \rho^2,\rho)$ is 
the constant function
$f_* \equiv \rho$. 

\section{Singularity of the entropy in a one dimensional example}\label{sec:example}

\begin{figure}
\begin{center}
\vskip-100pt
\includegraphics[scale=0.55]{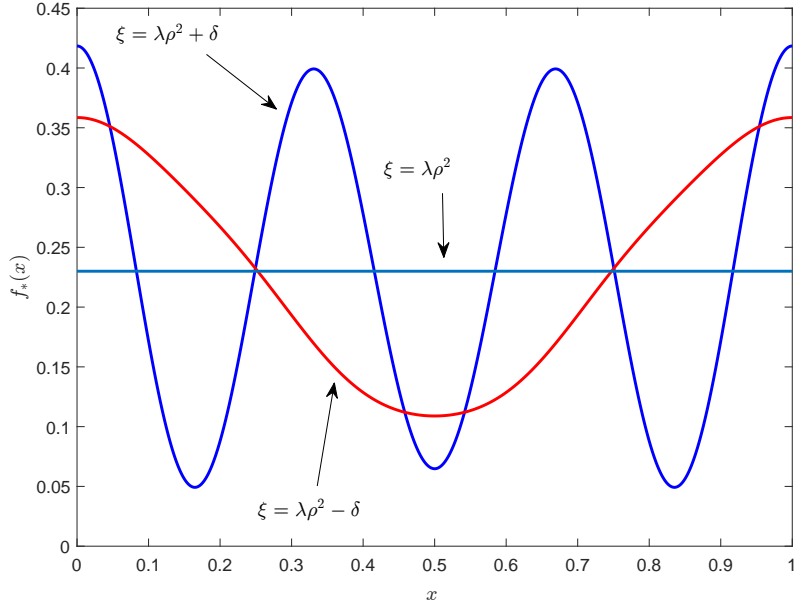}
\vskip-100pt 
\caption{Optimizers $f_*$ of $S(\xi,\rho)$ 
(computed numerically) when $\psi$ is given by Assumption~\ref{A2} with $r = 1/2$ and $M = 10$.  The plots 
show optimizers $f_*$ of $S(\xi,\rho)$ at 
fixed $\rho = 0.23$ and $3$ different 
values of $\xi$: on 
the transition curve ($\xi = \lambda \rho^2$) as well 
as just below ($\xi = \lambda \rho^2-\delta$) and just above 
($\xi = \lambda \rho^2+\delta$) the transition curve (here $\delta = 0.02$).}\label{fig2}
\end{center}
\end{figure}

Here we consider an 
example in dimension 
$d=1$ in which the entropy 
$S$ is singular. 
We will consider $\psi$ with 
the following modified 
power law structure.
\begin{assumption}\label{A2}
For some constants $r \in (0,1)$
and $M>0$,
\begin{equation*}
\psi(t) = \begin{cases} t^{-r}, 
& 0 < t < 1/4 \\ 
M, & 1/4 \le t \le 1/2\end{cases},
\end{equation*}
and $\psi(0) = 0$.
Also, $\psi$ is symmetric: 
for each $t \in [0,1]$, 
$\psi(t) = \psi(1-t)$. 
\end{assumption} 
Note that symmetry of $\psi$ 
corresponds to periodic 
boundary conditions for the 
particle configurations, {\em i.e.}, 
particle configurations on a circle. 
Clearly, $\psi$ satisfies Assumption~\ref{A1}.
If $M$ is chosen carefully, 
then 
at a given particle 
density, at 
high energy configurations 
tend to be multimodal, while 
at low energy 
configurations tend to 
be unimodal; see Figure~\ref{fig2}. 
The switch from unimodal 
to multimodal structure 
corresponds to a singularity 
in the entropy, as we 
show in Theorem~\ref{thm_trans} 
below. To make 
this argument rigorous, 
we need two ingredients. 
First, we identify 
where the interface between 
unimodal and multimodal structure 
should occur. 
The simplest guess is
that the interface corresponds 
to parameter values $(\xi,\rho)$ 
at which the optimizers are 
constant 
valued occupation densities $f \equiv \rho$. 
This guess turns out to be correct, 
as we show  
below. And second, we have to verify 
that parameters $(\xi,\rho)$ on both sides 
of this interface are achievable, 
so that the transition interface 
is in the interior of $\Omega$. We 
prove this in Lemma~\ref{lem_int} 
below. 

Before proceeding with the proof we introduce 
some notation. We write  
\begin{equation*}
\lambda = \int_{[0,1]^2} \psi(|x-y|)\,dx\,dy
\end{equation*}
for the integrated interaction function, 
and we define
\begin{equation*}
 \xi(f) = \int_{[0,1]^2} f(x)f(y)\psi(|x-y|)\,dx\,dy.
\end{equation*}
When $\xi = \lambda \rho^2$, 
the constant function 
$f \equiv \rho$ satisfies 
the constraints and is 
therefore an optimizer of 
the entropy. The curve 
$\xi = \lambda\rho^2$ is 
the interface between unimodal 
and multimodal optimizers discussed 
above, and it corresponds 
to a singularity in 
the microcanonical entropy, as we show in 
Theorem~\ref{thm_trans} below.
We now show this interface 
lies in the interior of $\Omega$, 
at least for a range of densities 
$\rho$.
\begin{lemma}\label{lem_int}
For each $r \in (0,1)$, 
there is an interaction $\psi$
satisfying Assumption~\ref{A2} 
with the following property. There is $\epsilon > 0$ 
such that the curve 
\begin{equation*}
\{(\xi,\rho) \in \Omega\,:\, \xi = \lambda\rho^2, 
\,\rho \in (1/4-\epsilon, 1/4)\}
\end{equation*}
is in the interior of $\Omega$.  
\end{lemma}

\begin{figure}
\begin{center}
\vskip-120pt
\includegraphics[scale=0.55]{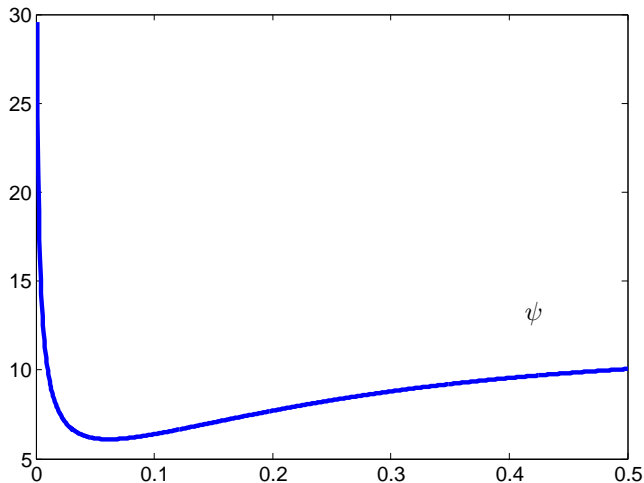}
\vskip-120pt
\caption{An example of a  
smooth interaction $\psi$ for 
which $S$ is singular as in 
Theorem~\ref{thm_trans}. 
Here we take 
periodic boundary conditions, 
{\em i.e.}, $\psi(t) = \psi(1-t)$.}\label{fig3}
\end{center}
\end{figure}

Lemma~\ref{lem_int} is proved by 
exhibiting functions $f$ which 
integrate to $\rho$ 
and have values of $\xi(f)$ 
both larger and smaller than $\xi(\rho) = \lambda\rho^2$. Such functions 
can be found for suitable $M$. 
We do not attempt to find 
the complete interior or boundary of $\Omega$. Fortunately, Lemma~\ref{lem_int} 
suffices for the following.

\begin{theorem}\label{thm_trans}
Let $\psi$ be as in 
Lemma~\ref{lem_int} with $r<1/2$. Then 
the entropy $S$ is non-differentiable 
along the curve $
\{(\xi,\rho) \in \Omega\,:\, \xi = \lambda\rho^2, 
\,\rho \in (1/4-\epsilon, 1/4)\}$.
\end{theorem}

Note that we needed Lemma~\ref{lem_int} to 
show that the curve $\xi = \lambda \rho^2$ 
is in the interior of $\Omega$ 
for $\rho \in (1/4-\epsilon, 1/4)$; 
otherwise, the notion of a singularity 
along $\xi = \lambda \rho^2$ does 
not make sense.
Theorem~\ref{thm_trans} 
shows there is a first order phase 
transition, {\em i.e.}, a discontinuity  
in the first derivative of the 
entropy, across 
this curve.
The curve 
corresponds to optimizers that 
are constant valued. 
(Recall from Corollary~\ref{cor_cts} 
that optimizers must be non-constant  
off this curve.) The first order 
transition corresponds to a qualitative
change in 
the non-constant optimizers 
across the singularity, 
namely, a change from unimodal 
to multimodal structure. 

We choose $\psi$ above for 
simple arguments. 
Though $\psi$ is not 
continuous at $1/4$ in general, 
it will be clear that 
the results above 
also hold for a smoothed
version of $\psi$; 
see the remarks below the 
proof of Lemma~\ref{lem_int} 
in Section~\ref{sec:appendix}. 
Indeed, modified 
versions of the arguments 
in the proofs below will 
go through for suitable bimodal 
potentials $\psi$ with a 
$\psi(t) = \psi(1-t)$, 
including potentials with 
a shape like the Lennard-Jones 
potential~\cite{LJ} on $[0,1/2]$. See Figure~\ref{fig3}.

\section{Discussion}\label{sec:discussion}

An interesting connection between 
certain random graph models and lattice 
statistical mechanics models is found in the 
Curie-Weiss mean field 
Ising model. Consider the 
case where $d = 1$ and 
$\psi \equiv J$ is constant. Then the energy density becomes 
\begin{equation}\label{CW}
n^{-2}J\sum_{i,j=1}^n \eta(i)\eta(j).
\end{equation}
This is the same as the pair interaction energy in 
the Curie-Weiss model in dimension $d=1$~\cite{kochmanski}. Now consider
a random graph model where a 
graph $X = (X_{ij})_{1\le i,j\le n}$ is represented by 
its adjacency matrix: $X_{ij} = 1$ if 
there is an edge from $i$ to $j$, 
and $X_{ij} = 0$ otherwise. 
$X$ can be directed or undirected; 
if it is undirected $X_{ij} = 1$ implies 
$X_{ji} = 1$ and vice-versa. 
The energy density of $X$ is defined as 
\begin{equation}\label{star_ct}
n^{-3}\sum_{i=1}^n \sum_{j,k=1}^n X_{ij}X_{ik}. 
\end{equation}
Note that $X_{ij}X_{ik} = 1$ if and only 
if there is an edge from $i$ to $j$ 
and from $i$ to $k$ in $X$. 
Thus, the energy can be considered a count 
of the number of $2$-stars embedded 
in $X$ (if $X$ is directed, the $2$-stars are outward directed). In addition to the energy density, a particle density is defined as 
\begin{equation}\label{PD}
n^{-2}\sum_{i=1}^n \sum_{j=1}^n X_{ij}.
\end{equation}
Directed and undirected versions 
of this model have been studied in both 
the grand canonical~\cite{aristoff_zhu1,radin_yin} 
and microcanonical~\cite{aristoff_zhu,kenyon,kenyon2} 
setting. In some cases, 
the outer sums (over $i$) in~\eqref{star_ct} and~\eqref{PD} can be ``decoupled'' 
from the inner sums (over $j,k$). 
The inner sums, 
namely $n^{-2}\sum_{j,k=1}^n X_{ij}X_{ik}$ 
and $n^{-1}\sum_{j=1}^nX_{ij}$ 
after appropriate normalization,
look like the 
Curie-Weiss energy and particle 
density
in dimension $d=1$ (when $i$ is considered fixed), and for this reason 
the relevant free energies and 
entropies 
of such random graph models are 
closely related to the corresponding 
quantities in the Curie-Weiss model. 
See~\cite{kochmanski} for a description and 
analysis of the Curie-Weiss model 
and~\cite{aristoff_zhu1,aristoff_zhu,
kenyon,kenyon2,radin_yin} for details 
and discussion on the above mentioned 
random 
graph models. Our models differ 
from such random
graph models in that the 
interaction 
between edges is allowed to depend on 
the distance between the edges.

Another way to view our models 
is as follows. When $d =2$, 
a configuration $\eta:\{1,\ldots,n\}^2 \to \{0,1\}$
corresponds to the adjacency 
matrix of a directed graph: 
$\eta(i,j) = 1$ if there 
is an edge from $i$ to $j$ 
and $\eta(i,j) = 0$ otherwise. 
In this case, 
$\psi$ corresponds to 
an interaction between edges.  
If $\psi$ is nonconstant, 
it introduces a underlying 
geometry to the graphs. 
For instance, if $\psi$ is 
repulsive and $\eta(i,j) = 1$, 
then for fixed 
particle density and 
sufficiently low energy density, 
other edges are not likely to appear 
``near'' the directed edge 
$(i,j)$. (For 
clarity, we 
have defined ``near'' 
in the context of the Euclidean norm. 
However, inspection of Lemma~\ref{lem_cts} 
and Lemma~\ref{lem_equiv} 
show that our main result, Theorem~\ref{thm_var}, continues to hold 
when the Euclidean 
norm $|\cdot|$ is replaced with any other 
norm.) To see how $\psi$ might 
capture geometric features of 
graphs, consider the case of 
a repulsive potential with 
a cutoff, and assume particle density 
is fixed. Graphs at low energy density 
likely have lower connectivity, 
since pairs of edges at distance 
less than the cutoff are not 
likely to appear together; 
on the other hand, graphs at high
energy density may tend to cluster. 
Thus, we expect that the energy density is 
related to clustering and connectivity properties of 
the graphs. From a statistical 
perspective, $\psi$ 
allows us to capture second as 
well as first order statistics of graphs, 
for instance, edge correlations 
as well as mean edge density. 

When $d=2$, the limiting 
occupation density is related to a certain 
type of graph limit called {\em graphon}
~\cite{lovasz}. 
Formally, a {graphon} is a symmetric 
measurable function $g:[0,1]^2 \to [0,1]$. 
Intuitively, graphons $g$ represent 
an edge probability density: namely,
$g(x,y)$ represents the probability 
for an edge between $x$ and $y$, 
where $x,y \in [0,1]$ lie on a 
continuum of vertices. Interestingly, 
it has recently been shown that in certain 
random 
graph models where the 
densities of edges and 
certain embedded subgraphs (for instance, 
$2$-stars, as discussed above) 
are held constant, the  
graphon $g_*$ that optimizes entropy 
tends to form facets, that is, 
$g_*$ is either constant 
or piecewise constant (up to 
a relabeling of vertices); see 
for instance~\cite{kenyon,radin_sadun,radin_sadun2, radin_ren_sadun}. In contrast, 
we have shown in Corollary~\ref{cor_cts} 
that our optimizers $f_*$ must 
be continuous (provided they 
satisfy the Euler-Lagrange 
equations~\eqref{EL_rearrange}). The reason seems 
to be that the geometry 
associated with $\psi$ enforces
some regularity on the structure 
of the optimizers.

\section{Proofs}\label{sec:appendix}

Before proceeding with the proofs, 
we introduce some terminology 
from large deviations theory. 
A family of probability 
measures ${\mathbb Q}_n$ 
on ${\mathcal T}$ is called 
{\em exponentially tight} 
if all compact subsets of ${\mathcal T}$ 
are measurable and for every $M<\infty$, 
there is a compact set $C \subset {\mathcal T}$ such that 
$\lim_{n\to \infty}n^{-1}\log {\mathbb Q}_n({\mathcal T}\setminus C) < -M$. 
Two families $X_n$, $Y_n$ of real-valued  random 
variables defined on the same 
probability space are called 
{\em exponentially equivalent} 
with speed $n^d$ 
if for each $\delta > 0$, 
the event $\{|X_n-Y_n|>\delta\}$ 
is measurable with
$\lim\sup_{n \to \infty} 
n^{-d}\log {\mathbb P}(|X_n-Y_n|>\delta) = -\infty$. Given a topological vector
space ${\mathcal T}$ over ${\mathbb R}$ and a rate function $K:{\mathcal T} \to {\mathbb R}$,  an {\em exposing 
hyperplane} 
for $y \in {\mathcal T}$ is 
an element $\lambda \in {\mathcal T}^*$ 
such that 
$\langle \lambda, y\rangle - K(y) > 
\langle \lambda,z\rangle - K(z)$ 
for all $z \in {\mathcal T}$ 
with $z \ne y$ (here ${\mathcal T}^*$ 
denotes the dual space of ${\mathcal T}$ 
and $\langle \cdot,\cdot\rangle: {\mathcal T}^* \times {\mathcal T} \to {\mathbb R}$ 
the natural pairing). See~\cite{dembo} for details.

We restate the {\em contraction principle} 
from large deviations theory for 
our purposes as 
follows. (See Theorem 4.2.1 of~\cite{dembo}.) Let ${\mathbb Q}_n$ 
be a family of probability 
measures on a Hausdorff 
topological space ${\mathcal T}$ 
satisfying a large deviations 
principle with speed $n^{d}$ 
and rate function 
$K:{\mathcal T} \to {\mathbb R}$. 
If $F:{\mathcal T} \to {\mathbb R}$ is {continuous}, then the 
family ${\tilde {\mathbb Q}}_n$ 
of pushforwards of ${\mathbb Q}_n$ 
by $F$ (defined by ${\tilde {\mathbb Q}}_n(A) = {\mathbb Q}_n(F^{-1}(A))$ 
for measurable $A \subset {\mathbb R}$) satisfies a large 
deviations principle with 
speed $n^d$ and rate function 
$L:{\mathbb R} \to {\mathbb R}$,
$$L(y) := \inf_{x \in {\mathcal T}\,:\,F(x) = y}K(x).$$

We begin by proving Theorem~\ref{thm_large deviation principle}. 
First, we need the following lemmas.
\begin{lemma}\label{lem_FL}
For any $s \in [0,1]$, 
\begin{equation}\label{FL}
\sup_{t \in {\mathbb R}} 
\left[st - \log\left(\frac{1}{2}+\frac{1}{2}e^t\right)\right] = H_{\textup{bin}}(s).
\end{equation}
\end{lemma}
\begin{proof}
When $s \notin [0,1]$ the quantity in 
brackets has no upper bound in $t$. When $s \in [0,1]$, the maximum is attained 
when $t = \log s - \log(1-s)$, 
and plugging this back into~\eqref{FL} yields the result.
\end{proof}

\begin{lemma}\label{lem_Hstar}
Suppose $\theta: {\mathbb R} \to 
{\mathbb R}$ is Lipschitz continuous 
and let $g \in L^p({\mathcal C})$. Then 
\begin{equation*}
\lim_{n \to \infty} 
n^{-d} \sum_{I \in \Lambda_n} 
\theta\left(n^d\int_{{\mathcal C}_I}g(x)\,dx\right) = \int_{\mathcal C} \theta(g(x))\,dx.
\end{equation*}
\end{lemma}
\begin{proof}
Consider 
the operator $A_n:L^q({\mathcal C}) 
\to L^q({\mathcal C})$ defined by 
\begin{equation*}
A_n g = \sum_{I \in \Lambda_n} 
\left(n^d \int_{{\mathcal C}_I} g(x)\,dx\right){1}_{{\mathcal C}_I}.
\end{equation*}
Note that
\begin{equation*}
\int_{\mathcal C}\theta(A_n g(x))\,dx = n^{-d}\sum_{I \in \Lambda_n} 
\theta\left(n^d\int_{{\mathcal C}_I}g(x)\,dx\right).
\end{equation*}
Since $\theta$ is Lipschitz, for a constant $c>0$,
\begin{align}\begin{split}\label{theta}
\left|\int_{\mathcal C}\theta(A_n g(x))\,dx 
- \int_{\mathcal C}\theta(g(x))\,dx\right|
&\le \int_{\mathcal C}|\theta(A_n g(x))
-\theta(g(x))|\,dx \\
&\le c\int_{\mathcal C} |A_n g(x)-g(x)|\,dx.\end{split}
\end{align}
Clearly $A_n g \to g$ in norm when $g$ is 
continuous. Since $A_n$ is a bounded 
operator and continuous functions 
are dense in $L^q({\mathcal C})$, 
we see that $A_n g \to g$ in norm for 
any $g \in L^q({\mathcal C})$. 
Thus,  
the last expression in~\eqref{theta} vanishes 
as $n \to \infty$.
\end{proof}

\begin{proof}[Proof of Theorem~\ref{thm_large deviation principle}]
Recall that 
\begin{align}\begin{split}\label{YF}
{\mathcal Y} &:= \left\{f \in {\mathcal X}\,:\, 
f(x) \in [0,1] \text{ for a.e. }x\right\},\\
{\mathcal F} &:= \left\{f \in {\mathcal X}\,:\, \exists \epsilon>0 \text{ s.t. }
f(x) \in [\epsilon,1-\epsilon] \text{ for a.e. }x\right\}.\end{split}
\end{align}
We claim that 
${\mathcal Y}$ is compact. 
Note that ${\mathcal Y}$ 
is closed, convex and 
bounded in $L^p({\mathcal C})$. 
Thus, by the Banach-Alaoglu theorem, 
${\mathcal Y}$ is 
compact if $1<p<\infty$. 
Since the weak topology in 
$L^1({\mathcal C})$ is 
coarser than the weak 
topology in $L^p({\mathcal C})$
for $1<p<\infty$, the $p=1$ 
case follows. 
We follow Baldi's theorem; see
Theorem 4.5.3 of~\cite{dembo}. 
Let ${\mathbb E}_n$ be 
expectation associated 
to ${\mathbb P}_n$. 
Write $f^{\eta_n}:{\mathcal C} \to [0,1]$ 
for the function drawn from 
${\mathbb P}_n$ associated 
to $\eta_n:\Lambda_n \to \{0,1\}$. 
Thus,
$(\eta_n(I))_{I \in \Lambda_n}$ are 
iid Bernoulli-$1/2$ random variables. 
For any
$g \in L^{q}({\mathcal C})$, 
\begin{align*}
H^*(g)&:=\lim_{n\to \infty} n^{-d} \log 
{\mathbb E}_n\left[\exp\left(
n^d \int_{{\mathcal C}} f^{\eta_n}(x)g(x)\,dx\right)\right] \\
&= \lim_{n\to \infty} n^{-d} \log 
{\mathbb E}_n\left[\exp\left(n^d
\sum_{I \in \Lambda_n} \eta_n(I) \int_{{\mathcal C}_{I}} g(x)\,dx\right)\right]
\\
&= 
\lim_{n\to \infty} n^{-d} 
\log \prod_{I \in \Lambda_n}  
{\mathbb E}_n\left[\exp\left(
\eta_n(I) n^d   \int_{{\mathcal C}_{I}} g(x)\,dx\right)\right] \\
&= 
\lim_{n\to \infty} n^{-d} 
 \sum_{I \in \Lambda_n} \log
{\mathbb E}_n\left[\exp\left(
\eta_n(I) n^d  \int_{{\mathcal C}_{I}} g(x)\,dx\right)\right]\\
&= \lim_{n\to \infty} n^{-d} 
 \sum_{I \in \Lambda_n} \log
\left[\frac{1}{2} + \frac{1}{2}\exp\left(n^d \int_{{\mathcal C}_I}g(x)\,dx\right)\right]\\
&=  \int_{\mathcal C}\log\left(\frac{1}{2}+\frac{1}{2}e^{g(x)}\right)\,dx.
\end{align*}
The last equality follows 
from Lemma~\ref{lem_Hstar}, since $\theta(t) := \log(\frac{1}{2}+\frac{1}{2}e^t)$ is Lipschitz.
Notice ${\mathcal Y}$ is 
compact and the ${\mathbb P}_n$ 
are supported on ${\mathcal Y}$. 
In particular, ${\mathbb P}_n$
is exponentially 
tight. Thus  
(see Theorem 4.5.3 (a) of~\cite{dembo}) 
${\mathbb P}_n$ satisfies the 
large deviations upper bound 
in $L^p({\mathcal C})$ with 
rate function
\begin{align*}
H(f)&= \int_{\mathcal C} H_{\textup{bin}}(f(x))\,dx \\
 &= \sup_{g \in L^q({\mathcal C})}
\left[\int_{\mathcal C}f(x)g(x)\,dx -
\int_{\mathcal C}\log\left(\frac{1}{2}+\frac{1}{2}e^{g(x)}\right)\,dx 
\right],
\end{align*}
where we used Lemma~\ref{lem_FL} 
for the second equality. 
Since the weak topology is 
coarser than the norm topology, 
${\mathbb P}_n$ also satisfies 
the large deviations upper 
bound in ${\mathcal X}$. 
It is easy to check that the rate function
$H$ is nonnegative 
and lower semi-continuous. 
We now verify the remaining 
conditions in Baldi's theorem.
Let $f \in {\mathcal F}$, and 
define $h_f(x) = \log f(x)- \log(1-f(x))$ 
for $x \in {\mathcal C}$. 
Then $h_f$ is an 
exposing hyperplane for $f$, since
\begin{align}\begin{split}\label{hyper}
&\int_{\mathcal C} f(x)h_f(x)\,dx- H(f) - 
\left[\int_{\mathcal C} g(x)h_f(x)\,dx- H(g) \right] \\
&= 
\int_{\mathcal C} \left(g(x)\log\frac{g(x)}{f(x)}+ (1-g(x))\log \frac{1-g(x)}{1-f(x)}\right)\,dx >0
\end{split}
\end{align}
whenever $f \ne g$ on a set 
of positive measure. Clearly, 
$H^*(h_f)$ exists and
$H^*(\gamma h_f)$ exists 
and is finite for 
all $\gamma > 1$. 
If for all 
open sets $U \subset {\mathcal X}$, 
we have
\begin{equation}\label{inf}
\inf_{f \in U \cap {\mathcal F}} H(f)
= \inf_{f \in U} H(f),
\end{equation}
then 
(see Theorem 4.5.20 (b)-(c) 
of~\cite{dembo}) 
${\mathbb P}_n$ satisfies 
the large deviations lower 
bound in ${\mathcal X}$. 
The norm topology in 
$L^p({\mathcal C})$ 
is coarser than the uniform 
topology, since an 
$\epsilon$-ball in $L^p({\mathcal C})$ 
contains the corresponding 
uniform 
$\epsilon$-ball when $\epsilon<1$. Thus,  
to prove~\eqref{inf} it 
suffices to consider a set
$U$ open in the uniform 
topology. If $U \cap {\mathcal F} 
= \emptyset$ then $H(f) = \infty$ 
for all $f \in U$ and both 
sides of~\eqref{inf} equal $\infty$. 
Suppose then that
$f \in U$ with $H(f)<\infty$, and define
\begin{equation*}
f_\epsilon(x) = \begin{cases} 
f(x),& f(x) \in [\epsilon,1-\epsilon]\\
1-\epsilon, & f(x) > 1-\epsilon\\
\epsilon, & f(x) < \epsilon\end{cases}.
\end{equation*}
For $\epsilon$ sufficiently 
small, $f_\epsilon \in U \cap {\mathcal F}$ and $H(f_\epsilon) \le H(f)$. 
This shows that 
\begin{equation*}
\inf_{f \in U \cap {\mathcal F}} H(f)
\le \inf_{f \in U} H(f).
\end{equation*}
The reverse inequality 
holds since $U \cap {\mathcal F} \subset U$, so we are done.
\end{proof}

Now we turn to the proof of 
Theorem~\ref{thm_var}. 
We will need Lemmas~\ref{lem_cts} 
and~\ref{lem_equiv} below. 
\begin{lemma}\label{lem_cts}
The maps ${\mathcal Y} 
\to {\mathbb R}$ defined by 
\begin{equation*}
f \mapsto \int_{\mathcal C^2} f(x)f(y)\psi(|x-y|)\,dx\,dy, \qquad 
f \mapsto \int_{\mathcal C} f(x)\,dx
\end{equation*}
are continuous.
\end{lemma}
\begin{proof}
Let $\{f_n\}$ 
in ${\mathcal Y}$ 
converge to $f \in {\mathcal Y}$. 
From Assumption~\ref{A1}, 
$(x,y)\mapsto \psi(|x-y|)$ 
is in $L^q({\mathcal C}^2)$. 
By Jensen's inequality, it follows that
$x \mapsto\int_{\mathcal C}\psi(|x-y|)\,dy$ is in $L^q({\mathcal C})$, since
\begin{equation*}
\int_{\mathcal C}\left(
\int_{\mathcal C} \psi(|x-y|)\,dy\right)^q\,dx 
\le \int_{{\mathcal C}^2} 
\psi(|x-y|)^q\,dx\,dy < \infty.
\end{equation*}
By boundedness of $f$, 
$x \mapsto \int_{\mathcal C} f(y)\psi(|x-y|)\,dy$ is also in $L^q({\mathcal C})$ and thus
\begin{equation}\label{ab1}
\lim_{n\to \infty} 
\int_{\mathcal C} [f_n(x)-f(x)]\left(\int_{\mathcal C} f(y)\psi(|x-y|)\,dy\right)dx = 0.
\end{equation}
Notice also that since
$(x,y)\mapsto \psi(|x-y|)\in L^q({\mathcal C}^2)$, 
$y \mapsto \psi(|x-y|)$ 
is in $L^q({\mathcal C})$.
So since $\psi$ is integrable and 
$f_n$, $f$ are uniformly bounded, by
dominated convergence 
\begin{equation*}
\lim_{n\to \infty} \int_{\mathcal C}\left|\int_{\mathcal C}[f_n(y)-f(y)]\psi(|x-y|)\,dy\right|\,dx = 0.
 \end{equation*}
Thus, using uniform boundedness of $f_n$ again,
\begin{equation}\label{ab2}
\lim_{n\to \infty} 
\int_{\mathcal C}  f_n(x)
\left[\int_{\mathcal C}[f_n(y)-f(y)]\psi(|x-y|)\,dy\right]dx = 0.
\end{equation}
Combining~\eqref{ab1} and~\eqref{ab2} yields
\begin{equation*}
\lim_{n\to \infty} 
\int_{\mathcal C^2} [f_n(x)f_n(y)-f(x)f(y)]\psi(|x-y|)\,dx\,dy = 0.
\end{equation*}
Continuity of the other map is 
clear, so the proof is complete.
\end{proof}

Next we prove exponential 
equivalences for the 
sums defining $N_n(\eta)$ and 
$E_n(\eta)$.
\begin{lemma}\label{lem_equiv}
For any $\epsilon>0$, 
\begin{equation}\label{supexp1}
\limsup_{n\to \infty}n^{-d}\log{\mathbb P}_n\left(\left|n^{-d}\sum_{I \in \Lambda_n} \eta_n(I) -
\int_{\mathcal C} f^{\eta_n}(x)\,dx\right| \ge \epsilon\right) = -\infty
\end{equation}
and
\begin{align}\begin{split}\label{supexp2}
&\limsup_{n\to \infty}n^{-d}\log{\mathbb P}_n\left(\left|n^{-2d}\sum_{I,J \in \Lambda_n} 
\eta_n(I) \eta_n(J) \phi_n(I,J) \right.\right. \\
&\qquad\qquad\qquad\qquad \qquad-\left.\left. 
\int_{\mathcal C^2} f^{\eta_n}(x)f^{\eta_n}(y)\psi(|x-y|)\,dx\,dy \right| \ge \epsilon\right) = -\infty.
\end{split}
\end{align}
\end{lemma}

\begin{proof}
By the definitions of $\eta_n$ and $f^{\nu_n}$,
\begin{equation*}
n^{-d}\sum_{I \in \Lambda_n} \eta_n(I) = 
\int_{\mathcal C} f^{\eta_n}(x)\,dx,
\end{equation*}
which implies~\eqref{supexp1}.
Define 
\begin{equation*}
\phi_n^{I,J} = \int_{{\mathcal C}_I \times {\mathcal C}_J} \psi(|x-y|)\,dx\,dy
\end{equation*}
and observe that
\begin{align}\begin{split}\label{Riemann}
&\left|n^{-2d}\sum_{I,J \in \Lambda_n} 
\eta_n(I) \eta_n(J) \phi_n(I,J) - 
\int_{{\mathcal C}^2} f^{\nu_n}(x)f^{\nu_n}(y)\psi(|x-y|)\,dx\,dy \right| \\
&= \left|\sum_{I,J \in \Lambda_n} 
\eta_n(I) \eta_n(J)\left(n^{-2d}\phi_n(I,J)
-\phi_n^{I,J}\right) \right| \\
&\le \sum_{I,J \in \Lambda_n}
\left|n^{-2d}\phi_n(I,J)-\phi_n^{I,J}\right|.\end{split}
\end{align}
Using Riemann integrability of 
$(x,y)\mapsto \psi(|x-y|)$, it 
is easy to see that the 
last expression in~\eqref{Riemann} 
is less than $\epsilon$ for 
sufficiently large $n$. This 
implies~\eqref{supexp2}.
\end{proof}

\begin{proof}[Proof of Theorem~\ref{thm_var}]
This is an immediate consequence of 
the contraction principle~\cite{dembo} along with 
Theorem~\ref{thm_large deviation principle}, Lemma~\ref{lem_cts}, 
and Lemma~\ref{lem_equiv}. 
\end{proof}

Now we are ready to prove Theorem~\ref{thm_EL}.
\begin{proof}[Proof of Theorem~\ref{thm_EL}]
Since ${\mathcal Y}$ is compact and
 \begin{equation*}
 f\in {\mathcal Y} \mapsto \int_{{\mathcal C}^2} 
 f(x)f(y)\psi(|x-y|)\,dx\,dy, \qquad 
 f\in {\mathcal Y} \mapsto \int_{\mathcal C} f(x)\,dx
 \end{equation*} 
 are continuous, ${\mathcal Y}_{\xi,\rho}$ 
 is compact. Thus, optimizers of~\eqref{var_prob} exist when 
 $(\xi,\rho)\in \Omega$. 
 Suppose now that $f_* \in {\mathcal F}$
 is an optimizer of~\eqref{var_prob} 
 for some $(\xi,\rho)$. 
For the remainder of the proof we will 
equip ${\mathcal Y}$ 
with the topology induced by the uniform norm. Thus, $f_*$ is in the 
interior of ${\mathcal Y}$.
To obtain the Euler-Lagrange 
equations~\eqref{EL_rearrange} 
we follow Theorem 9.1 of~\cite{clarke}. 
The multiplier rule there states
that there exist $\beta, \mu \in {\mathbb R}$ and $\nu \in \{0,1\}$ 
such that $(\beta,\mu,\nu) \ne (0,0,0)$ and for $f = f_*$ and all $\delta f \in L^{\infty}({\mathcal C})$,
\begin{align}\begin{split}\label{EL}
0 &= \beta \int_{\mathcal C}\left(\int_{\mathcal C} f(y)\psi(|x-y|)\,dy\right)\delta f(x)\,dx \\
&\qquad\qquad + 
\mu \int_{\mathcal C} \delta f(x)\,dx 
- \nu\int_{\mathcal C} H_{\textup{bin}}'(f(x))\delta f(x)\,dx,
\end{split}
\end{align}
provided the Frech{\'e}t 
derivatives in~\eqref{EL} are continuous 
for $f \in {\mathcal F}$. 
Continuity 
of the second Frech{\'e}t derivative 
is obvious. Continuity of
the first Frech{\'e}t derivative 
follows from integrability of $(x,y) \mapsto \psi(|x-y|)$, and continuity of the 
third Frech{\'e}t derivative follows from uniform continuity of  
$H_{\textup{bin}}'$ on $[\epsilon,1-\epsilon]$ for each $\epsilon>0$. Thus,
\begin{equation}\label{EL2}
\beta \int_{\mathcal C}f_*(y)\psi(|x-y|)\,dy 
+ \mu - \nu H_{\textup{bin}}'(f_*(x)) = 0
\end{equation}
for a.e. $x$. When $\nu = 1$, this is 
a rearrangement of~\eqref{EL_rearrange}. 
If $\nu = 0$ then $\beta \ne 0$ and
\begin{equation*}
\int_{\mathcal C} f_*(y)\psi(|x-y|)\,dy = \gamma
\end{equation*}
for a.e. $x$, where $\gamma = -\mu/\beta$. Note that
\begin{equation*}
\xi = \int_{{\mathcal C}^2} f_*(x)f_*(y)\psi(|x-y|)\,dy\,dx = \gamma \int_{\mathcal C} f_*(x)\,dx = \gamma \rho,
\end{equation*}
so in fact $\gamma = \xi/\rho$.
\end{proof}

\begin{proof}[Proof of Corollary~\ref{cor_cts}]
Let $f$ satisfy~\eqref{EL_rearrange}. Since $\gamma(t):= \psi(|t|)$ is locally integrable,
\begin{equation*}
(f \ast \gamma)(x) \equiv \int_{\mathcal C} f(y)\psi(|x-y|)\,dy
\end{equation*}
is continuous. Now continuity of 
$f$ follows from the Euler-Lagrange 
equation~\eqref{EL_rearrange}. 
\end{proof}

\begin{proof}[Proof of Lemma~\ref{lem_int}]
Let 
$0 < \rho \le 1/4$ and define 
\begin{equation*}
f_1(x) = 1_{[0,\rho]}(x), \quad 
f_2(x) \equiv \rho,\quad \text{and} 
\quad f_3(x) = 1_{[0,\rho/2]}+1_{[1/2-\rho/2,1/2]}.
\end{equation*}
Then $\rho = \int_{[0,1]} f_i(x)\,dx$ 
for $i=1,2,3$, and 
\begin{align}\begin{split}\label{ineq}
&\xi(f_1) = \frac{2\rho^{2-r}}{(1-r)(2-r)}, \\
&\xi(f_2) = \lambda \rho^2 = 2\frac{4^{r-1}\rho^2}{1-r}+ \frac{M\rho^2}{2}\\
&\xi(f_3) = 
\frac{4\left(\frac{\rho}{2}\right)^{2-r}}{(1-r)(2-r)} + \frac{M\rho^2}{2}.
\end{split}
\end{align}
Observe that when $\rho = 1/4$, 
\begin{equation}\label{ineq2}
\xi(f_2) = \frac{4^{r-5/2}}{1-r} + \frac{M\rho^2}{2} 
<  \frac{4^{3r/2-2}}{(1-r)(2-r)}+\frac{M\rho^2}{2}= \xi(f_3),
\end{equation}
where the inequality can be 
checked by straightforward calculus. 
If $\epsilon > 0$ is sufficiently 
small, $\xi(f_2)<\xi(f_3)$ 
whenever $\rho \in (1/4-\epsilon,1/4)$. 
Moreover, when $M$ is sufficiently large, 
$\xi(f_1) < \xi(f_2)$. 
All values of $\xi$ 
between $\xi(f_1)$ and $\xi(f_3)$ 
are attainable by, for example, 
taking convex combinations of 
$f_1$ and $f_3$. 
\end{proof}

To see that Lemma~\ref{lem_int} 
holds for a smoothed function 
of $\psi$, let $\gamma$ 
be a bounded function supported 
on $(1/4-\delta,1/4+\delta)$ 
such that $\gamma+\psi$ is 
smooth. Then for sufficiently 
small $\delta > 0$, the 
arguments above still go through.

\begin{proof}[Proof of Theorem~\ref{thm_trans}] Let 
$\psi$ be as in Lemma~\ref{lem_int}. 
Take $\rho \in (1/4-\epsilon,1/4)$,  
and let $f$ be an optimizer of~\eqref{var_prob} at 
$(\xi,\rho) \in \textup{int}(\Omega)$. 
We can write $
f(x) = \rho + \delta f(x)$, 
where
\begin{equation}\label{inteq0} 
\int_{[0,1]} \delta f(x)\,dx = 0.
\end{equation}   
Observe that
\begin{align*}
H(f)-H(\rho) &=\int_{[0,1]} 
H_{\textup{bin}}(\rho+\delta f(x))\,dx - H_{\textup{bin}}(\rho) \\
&= \int_{[0,1]} \left[H_{\textup{bin}}(\rho+\delta f(x)) - H_{\textup{bin}}'(\rho)\delta f(x) - H_{\textup{bin}}(\rho)\right]\,dx \\
&\ge c \int_{[0,1]}\delta f(x)^2 \,dx,
\end{align*}
where by convexity,
\begin{equation*}
c = \min_{t\in [-\rho,1-\rho]\setminus\{0\}} \frac{H_{\textup{bin}}(\rho+t)-H_{\textup{bin}}'(\rho)t-H_{\textup{bin}}(\rho)}{t^2} >0.
\end{equation*}
Note that $H(\rho) \le H(f)$ 
with equality if and only if 
$f \equiv \rho$ a.e. It follows that the optimizer 
of~\eqref{var_prob} at 
$(\lambda \rho^2,\rho)$ 
is the constant function 
with value $\rho$. Thus,
\begin{align*}
\delta S := S(\xi,\rho)-S(\lambda\rho^2,\rho)
 &= H(\rho)-H(f)
 \\
 &\le -c \int_{[0,1]}\delta f(x)^2 \,dx.
\end{align*}
Now note that
\begin{align}\begin{split}\label{xi}
\delta \xi &:= \xi(f)-\xi(\rho) \\
&= 
\int_{[0,1]^2}
(\rho+\delta f(x))(\rho+\delta f(y))
\psi(|x-y|)\,dx\,dy -
 \int_{[0,1]^2}
\rho^2 \psi(|x-y|)\,dx\,dy \\
&= 
2\rho\int_{[0,1]^2} 
\delta f(x)\psi(|x-y|)\,dx\,dy 
+ \int_{[0,1]^2} 
\delta f(x)\delta f(y)\psi(|x-y|)\,dx\,dy 
\\
&= \int_{[0,1]^2} 
\delta f(x)\delta f(y)\psi(|x-y|)\,dx\,dy,
\end{split}
\end{align}
with the last equality coming 
from~\eqref{inteq0} and the fact that for each $x \in [0,1]$,
\begin{equation*}
\int_{[0,1]} \psi(|x-y|)\,dy = \lambda.
\end{equation*}
Since $r<1/2$, the 
integral kernel $\Psi$ defined by
\begin{equation*}
\Psi f(y) = \int_{[0,1]}f(x)\psi(|x-y|)\,dx
\end{equation*}
is a Hilbert-Schmidt operator 
on $L^2[0,1]$. Thus,
\begin{equation*}
\left|\int_{[0,1]^2} 
\delta f(x)\delta f(y)\psi(|x-y|)\,dx\,dy\right| \le \sigma\int_{[0,1]}  \delta f(x)^2\,dx,
\end{equation*}
where $\sigma$ is the 
spectral radius of $\Psi$. 
Putting this in~\eqref{xi} yields
\begin{equation*}
|\delta \xi| \le \sigma\int_{[0,1]}  \delta f(x)^2\,dx.
\end{equation*}
Combining the estimates for $\delta S$ 
and $\delta \xi$, we get
\begin{equation*}
\delta S \le -\frac{c}{\sigma}|\delta \xi|.
\end{equation*}
Thus, 
$S$ is not differentiable 
at $(\lambda\rho^2,\rho) \in \Omega$.
\end{proof}

\section*{Acknowledgements}

The authors would like to thank 
R. Mark Bradley, 
Olivier Pinaud, Dan Pirjol, Charles Radin and 
Clayton Shonkwiler for helpful comments and suggestions. D.~Aristoff gratefully acknowledges support from the National Science Foundation via the award
NSF-DMS-1522398. L. Zhu gratefully acknowledges support from the National Science Foundation via the award
NSF-DMS-1613164. 
The authors 
would also like to thank 
an anonymous referee for a 
careful reading of the article 
and helpful suggestions.


%
%
%
%

\end{document}